\documentclass[10pt,fleqn]{article}

\usepackage[latin1]{inputenc}
\usepackage[T1]{fontenc}

\usepackage{xspace}
\usepackage{amsmath,amssymb,mathtools,dsfont}
\usepackage[pagebackref]{hyperref}
\usepackage{amsthm}
\usepackage{xifthen}
\usepackage{fullpage}

\usepackage[numbers]{natbib}

\usepackage{tikz}
\usetikzlibrary{decorations,arrows,petri,topaths,backgrounds,shapes,positioning,fit,calc,decorations.pathreplacing,patterns,intersections}

\usepackage{cleveref}

\newtheorem{theorem}{Theorem}
\crefname{table}{Table}{Tables}
\crefname{figure}{Figure}{Figures}
\crefname{theorem}{Theorem}{Theorems}
\crefname{corollary}{Corollary}{Corollaries}
\crefname{observation}{Observation}{Observations}
\crefname{lemma}{Lemma}{Lemmas}
\crefname{example}{Example}{Examples}
\crefname{reduction}{Reduction}{Reductions}
\crefname{construction}{Construction}{Constructions}
\crefname{subsection}{Subsection}{Subsections}
\crefname{section}{Section}{Sections}

\theoremstyle{definition}
\newtheorem{definition}{Definition}
\crefname{definition}{Definition}{Definitions}

\newcommand{\dspace}{\ensuremath{\mathds{R}^{d}}}

\newcommand{\ppp}{{\cal P}}
\newcommand{\vvv}{{\cal V}}

\newcommand{\aaa}{{\cal A}}
\newcommand{\rrr}{{\cal R}}

\newcommand{\dde}[1][]{{\ifthenelse{\equal{#1}{}}{$d$}{$#1$}-dimensional Euclidean}\xspace}
\newcommand{\dder}[1][]{{\ifthenelse{\equal{#1}{}}{$d$}{$#1$}-dimensional Euclidean embedding}\xspace}

\newcommand{\dEuclid}[1][]{{\ifthenelse{\equal{#1}{}}{$d$}{$#1$}-Euclidean}\xspace}

\newcommand{\dis}[1]{\ensuremath{\|#1 \|}}

\newcommand{\pref}{\ensuremath{\succ}}

\newcommand{\pos}{\ensuremath{E}}

\newcommand{\rank}{\ensuremath{\mathsf{rk}}}

\newcommand{\myemph}[1]{{\color{green!30!black}\emph{#1}}}

\newcommand{\vect}[1]{\ensuremath{\boldsymbol{#1}}}
\newcommand{\px}{\ensuremath{\vect{x}}}
\newcommand{\pp}{\ensuremath{\vect{p}}}
\newcommand{\pq}{\ensuremath{\vect{q}}}

\begin{document}
\newcommand{\mytitle}{2-Dimensional {E}uclidean Preferences}%

\title{\mytitle}

\author{Laurent Bulteau$^1$ \and Jiehua Chen$^2$}
\date{\small $^1$LIGM, CNRS, Universit{\'e} Gustave Eiffel, France\\
$^2$ TU Wien, Austria}
\maketitle

\begin{abstract}
  A preference profile with $m$~alternatives and $n$~voters is $2$-dimensional Euclidean 
  if both the alternatives and  the voters can be placed into a $2$-dimensional space 
  such that for each pair of alternatives, every voter prefers the one  
  which has a shorter Euclidean distance to the voter.
  We study how $2$-dimensional Euclidean preference profiles depend on the values~$m$ and $n$.
  We find that any profile with at most two voters or at most three alternatives is $2$-dimensional Euclidean
  while for three voters, we can show this property for up to seven alternatives.
  The results are tight in terms of Bogomolnaia and Laslier~\cite[Proposition 15(1)]{BoLa2006}.
\end{abstract} 

\section{Introduction}\label{sec:intro}
Using spatial models to analyze a preference profile (i.e., voters' preferences over alternatives) is a popular approach in political and social sciences, economics, and psychology. 
Here, the idea is to consider the voters and alternatives as points in a $d$-dimensional Euclidean space such that a voter prefers one alternative~$x$ over the other alternative~$y$ if and only if her Euclidean distance to~$x$ is smaller than to~$y$.
Such preference profiles are referred to as \myemph{\dEuclid} preference profiles~\cite{Coombs1964,BoLa2006}. 
Clearly, not every preference profile is Euclidean.
For $d=1$, checking whether a given preference profile is Euclidean can be done in polynomial time~\cite{DoiFal1994,Knoblauch2010,ElkFal2014}, 
while it is beyond NP for every fixed $d\ge 2$~\cite{Peters2017}.
In terms of characterization using forbidden subprofiles, Chen et al.~\cite{ChePruWoe2017} show that \dEuclid[1] preference profiles cannot be characterized via finitely many finite forbidden subprofiles.
Nevertheless,  Chen and Grottke~\cite{ChenGrottke2021} show that for two voters, a preference profile is \dEuclid[1] if and only if it is single-peaked, 
and for five voters,  single-peakedness and single-crossing are equivalent to 1-Euclideanness.
Since single-peaked and single-crossing preferences can be characterized via forbidden subprofiles~\cite{BH11,BCW12},
\dEuclid[1] preference profiles with up to five voters as well.
Bogomolnaia and Laslier show that a preference profile is \dEuclid\ if
the number of voters is at most $d$ or the number of alternatives is at most $d+1$.
Moreover, they provide a preference profile with $d+1$ voters and $2^{d+1}$ alternatives which is not \dEuclid and show that every preference profile with $d+1$ voters and $d+2$ alternatives is \dEuclid.
Notice that there is a gap in the size of \dEuclid\ and non-\dEuclid preference profiles.
Aiming to close this gap, we search for minimally non-\dEuclid[2] profiles. %
We obtain that a preference profile with $n$ voters and $m$~alternatives is always \dEuclid[2] if and only if $n\le 2$, or $n\le 3$ and $m \le 7$.
See \cref{fig:intervals} for an overview.

The paper is organized as follows:
\cref{sec:defi} introduces necessary definitions and notations.
In \cref{sec:2d}, we give explicit embeddings (i.e., points for the voters and the alternatives) for showing that any preference profile with $2$ voters or $3$ alternatives is \dEuclid[2]; note that this result has already been obtained by Bogomolnaia and Laslier~\cite{BoLa2006}.
In \cref{sec:exper}, we run our own heuristic and show that any preference profile with three voters and six alternatives is \dEuclid[2].
We conclude with future research directions in \cref{sec:conclude}. 

\section{Definitions and notations}
\label{sec:defi}

Given a non-negative integer~$t$, we use \myemph{$[t]$} to denote the set~$\{1,\ldots,t\}$.
Let $\aaa\coloneqq \{a_1,\ldots,a_m\}$ be a set of $m$~alternatives 
and let $\vvv\coloneqq \{v_1,\ldots,v_n\}$ be a set of $n$~voters. 
A \myemph{preference profile}~$\ppp\coloneqq (\aaa, \vvv, \rrr=(\pref_1, \ldots, \pref_n))$ specifies the \myemph{preference order}s of the voters in $\vvv$, 
where voter~$v_i$ ranks the alternatives according
to a strict linear order~$\pref_i$ over $A$. 
For alternatives~$a$ and $b$, the relation
$a\pref_i b$ means that voter~$v_i$ strictly prefers $a$ to $b$. $\{a,
b\}\pref_i c$ means that voter~$v_i$ strictly prefers $a$ and $b$ to $c$, but
the preference relation between alternatives~$a$ and $b$ is arbitrary but
unique. We also assume that all voters in a preference profile have pairwise
different preference orders.

For three distinct real values~$x, y, z$, by $x < \{y, z\}$ we mean that~$x <
y$ and $x < z$, and by $\{x, y\} < z$ we mean that $x < z$ and $y < z$.

By convention, we use bold~$\px$ to denote a vector or a point in a Euclidean space, and we use $\px[i]$ to refer to the $i^{\text{th}}$~value in~$\px$.
Given two points~$\pp, \pq$ in a $d$-dimensional space~$\dspace$, we write
\myemph{$\dis{\pp-\pq}_2$} to denote the Euclidean distance between $\pp$ and $\pq$, that is, 
$\dis{\pp-\pq}_2 \coloneqq \sqrt{\sum_{i\in [d]}(\pp[i]-\pq[i])^2}$.

\paragraph{\dder{s}.}
Generally speaking, the Euclidean representation models the preferences of voters over the alternatives 
using the Euclidean distance between an alternative and a voter. A
shorter distance indicates a stronger preference.
\begin{definition}[\dder{s}]
  Let $\ppp \coloneqq (\aaa\coloneqq \{a_1, \ldots, a_m\}$, $\vvv\coloneqq \{v_1, \ldots, v_n\}$, $\rrr\coloneqq(\pref_1, \ldots, \pref_n))$ be a
  preference profile. Let~$\pos\coloneqq \aaa \cup \vvv \to \dspace$ be a function. %
  A voter~$v_i \in \vvv$ is \emph{\dde} with respect to~$\pos$
  if for any two distinct alternatives~$a, b \in \aaa$, it holds that 
  \[a \pref_i b \text{ if and only if } \dis{\pos(a) - \pos(v_i)}_2 < \dis{\pos(b) - \pos(v_i)}_2. \]
  
  Function~$\pos$ is a \emph{\dder} of profile~$\ppp$ if each voter of~$\vvv$ is
  $d$-dimensional Euclidean with respect to $\pos$.
  A profile is \emph{d-Euclidean} if it admits a \dder.
\end{definition}

\begin{figure}[t]
  \centering \def\sep{0.4}
  \begin{tikzpicture}[draw=black!70, xscale=0.25, yscale=-0.25]
   \begin{scope}
     \node[] at (-4,-2){Worst-case}; 
     \node[] at (-4,0){dimension}; 
     
     \node[] at (18,-2){Number of Alternatives $(m)$};
     \node[] at (-4,6){Number};  \node[] at (-4,8){of voters};    \node[] at (-4,10){$(n)$};
     
     \foreach \nc in {1,...,8} \node[] at (\nc*2,0) {$\nc$}; 
     \foreach \x/\nc in {9.5/\cdots,11/15,12/16, 13.5/\cdots} \node[] at (\x*2,0) {$\nc$};       
     \foreach \nv in {1,...,6} \node[] at (0,\nv*2) {$\nv$}; 
     \node[] at (0,14) {$\vdots$}; 
     
     \draw (1,-3)--(1,15);
     \draw (-8,1)--(30,1);

     \tikzset{ every path/.style={draw=black, color=black,line width=2pt}}
     
     \tikzset{ }
     \node at (4,4) {1D};
     \draw (5,15) -- (5,5) -- (7,5) -- (7,3) -- (30,3);      
     \node[] at (6,6) {$\bullet$};
     \node[] at (8,4) {$\bullet$};
     
     \tikzset{ every path/.style={draw=blue, color=blue,line width=2pt}}
     \node at (8,6) {2D};       
     \draw (7,15) -- (7,7) -- (15,7) -- (15,5) -- (30,5); %
     \node[] at (8,8) {$\bullet$};
     \node[] at (16,6) {$\bullet$};    
     
     \tikzset{ every path/.style={draw=red, color=red,line width=2pt}}
     \node at (12,8) {3D};     
     \node at (19,8) {{\bf ?} (3 or 4)};     
     \draw (9,15) -- (9,9) -- (15,9) -- (23,9) -- (23,7) -- (30,7);  \draw (15,9) -- (15,7) -- (23,7) -- (30,7); 
     \node[] at (10,10) {$\bullet$};
     \node[] at (24,8) {$\bullet$};

     \tikzset{ every path/.style={draw=black, color=black,line width=2pt}}
 
     \node at (19,12) {$\geq$4D};     
     
   \end{scope}
  
   \end{tikzpicture}
  \caption{Boundaries of non-\dEuclid\ profiles with a given number of voters and alternatives. Each colored bullet point denotes the existence of such a non-Euclidean profile for the corresponding dimension.} %
  \label{fig:intervals}
\end{figure}
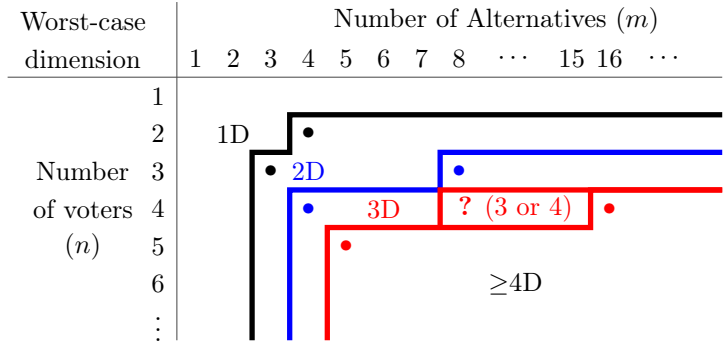

\section{Small \dEuclid[2] preference profiles}
\label{sec:2d}

In this section, we show two positive results, namely that any preference profile with at most two voters or three alternatives is \dEuclid[2].

\begin{theorem}
  Any preference profile with two voters is \dde[2].
\end{theorem}
\begin{proof}

  Let $\ppp \coloneqq (\aaa\coloneqq \{a_1, \ldots, a_m\}, \vvv\coloneqq \{v_1, v_2\}, \rrr\coloneqq (\pref_1, \pref_2))$ be a preference profile with two voters.
  We define an embedding as follows:
  Let $\pos(v_1)\coloneqq (-m^2,0)$, $\pos(v_2)\coloneqq(0,-m^2)$. %
  For each alternative $a\in \aaa$, 
  let \myemph{$\rank_1(a)$} (resp. \myemph{$\rank_2(a)$}) be the position of $a$ for voter $v_1$ (resp.\ $v_2$), i.e.,
  $\rank_i(a) \coloneqq |\{b\mid b \pref_i a\}|$ for $i\in [2]$.
  
  Define $\pos(a)\coloneqq (\rank_1(a), \rank_2(a))$.
  We now show that $\pos$ is a \dder[2] of $\ppp$.  
  To this end, we prove that for voter $v_1$ and an arbitrary alternative~$a\in \aaa$,
  we have
  \begin{align*}
    m^2+\rank_1(a) \leq\dis{\pos(a) - \pos(v_1)}_2< m^2+\rank_1(a)+1. 
  \end{align*}
 The lower bound is directly obtained by construction. 
 For the upper bound, we infer:
 \begin{align*}
   (m^2+\rank_1(a)+1)^2- \dis{\pos(a) - \pos(v_1)}^2
   &= (m^2+\rank_1(a)+1)^2 - \left((m^2+\rank_1(a))^2+ (\rank_2(a))^2\right) \\
   & = 2(m^2+\rank_1(a)) + 1 - (\rank_2(a))^2\\
   & > 0.
 \end{align*}
 The last inequality holds since $\rank_2(a) \le m$.
 Hence for each two alternatives $a$ and $b$ with $a \succ_1 b$, we have 
 \[\dis{\pos(a) - \pos(v_1)}_2<m^2+\rank_1(a)+1 \leq m^2+\rank_1(b) \leq \dis{\pos(b) - \pos(v_1)}_2.\]
 Thus $v_1$ is \dde[2] wrt.\ $\pos$.
 By symmetry of the construction, $v_2$ is also \dde[2] wrt.\ $\pos$.
\end{proof}

\begin{theorem}
  Any profile with at most three alternatives is \dde[2].
\end{theorem}
\begin{proof}
 In fact, any set of non-degenerated coordinates assigned to the alternatives yields a feasible solution for the voters.
 Assume that $\aaa=\{a,b,c\}$. Let $\pos(a)=(0,2)$, $\pos(b)=(2,-1)$ and $\pos(c)=(-2,-1)$.
 Then consider each voter $v_i$ and assign $v_i$ to a point $\pos(v_i)$ according to its preference order:
 \begin{itemize}
  \item[--] If $a\pref_i b \pref_i c$, then let $\pos(v_i)\coloneqq (2,2)$.
  \item[--] If $b\pref_i a \pref_i c$, then let $\pos(v_i)\coloneqq (2,0)$.
  \item[--] If $b\pref_i c \pref_i a$, then let $\pos(v_i)\coloneqq (1,-1)$.
  \item[--] If $c\pref_i b \pref_i a$, then let $\pos(v_i)\coloneqq (-1,-1)$.
  \item[--] If $c\pref_i a \pref_i b$, then let $\pos(v_i)\coloneqq (-2,0)$.
  \item[--] If $a\pref_i c \pref_i b$, then let $\pos(v_i)\coloneqq (-2,2)$.
 \end{itemize}
It is a simple matter of computation to verify that this is indeed a \dder[2]. For example, with $v_i\colon a\pref_i b \pref_i c$, i.e., $\pos(v_i)=(2,2)$, we have $\dis{\pos(a) - \pos(v_i)}_2=2$, then $\dis{\pos(b) - \pos(v_i)}_2=3>2$, and finally $\dis{\pos(c) - \pos(v_i)}_2=5>3$. 
 \end{proof}

 \section{Existence of \dEuclid[2] profiles by experiments}\label{sec:exper}
 In this section, we show the following positive result by empirical study.
 
\begin{theorem}
  Any profile with at most three voters and at most seven alternatives is \dde[2].
\end{theorem}

\begin{proof}
  We present a greedy heuristic algorithm used to determine that all profiles with at most $3$ voters and at most $7$ alternatives are \dEuclid[2].
  Given a profile as input, the algorithm tries to create a 2-dimensional embedding of this profile. However, it may fail even when the input profile is actually \dEuclid[2].
  Following a similar line as in the work of Chen and Grottke~\cite{ChenGrottke2021}, we did some optimization to significantly shrink the search space on all profiles: We only consider profiles with exactly three distinct preference orders over exactly seven alternatives since the Euclidean property is monotone, and we assume that one of the preference orders is~$1 \succ \cdots \succ 7$.
   The number of relevant profiles with $3$ voters and $7$ alternatives is $12693241$.
  
  The heuristic is particularly simple: We first place the three voters randomly in a $2$-dimensional space. 
  Then, we pick each alternative successively (in random order), and compute the ``free area'', that is the set of points at which this alternative may be placed in order to satisfy the ranking of each voter with respect to earlier alternatives.
  If the free area is non-empty, then we randomly choose a point inside it to embed this alternative. Otherwise, we quit this branch and start over until a maximum number of iterations is reached. 

  The free area itself is computed as an intersection of $n$ annuli: each voter indeed has a lower- and upper-bound on the distance to each new alternative. It is straightforward to decide, given any point, whether it belongs to the free area or not. Also, the intersection points of the different circles, creating the ``corners'' of the free area, can easily be computed. Hence we choose a random point in the free area by trying successive points in the smallest disk containing the corner points. 
All generated profiles, together with their \dEuclid[2] embeddings and the distances used for the verification, are available online at \url{https://owncloud.tuwien.ac.at/index.php/s/DqjD6OxQFpI6S0u}.
\end{proof}

\section{Conclusion}\label{sec:conclude}
This work opens up several future research directions.
First, it would be interesting to have a combinatorial algorithm for finding a \dEuclid[2] embedding for any profile with $3$ voters and $7$ alternatives.
Second, following the research question by Chen and Grottke~\cite{ChenGrottke2021} for the one-dimensional case,
it would be interesting to characterize \dEuclid[2] preference profiles with up to $3$ voters. 
Last, it remains to close the gap for the $3$-dimensional case, i.e., determining the minimum number~$m$ of alternatives for which there is a non-\dEuclid[3] preference profiles with $4$ voters.

\bibliographystyle{abbrvnat}

\bibliography{bib}

\end{document}